\documentclass[12pt]{article}

\usepackage{authblk}

\usepackage[margin=1in]{geometry}
\usepackage{amsmath,amsfonts,graphicx,framed}
\usepackage{amssymb}
\usepackage{fancyhdr}
\usepackage{float}
\usepackage{braket}
\usepackage{caption}
\usepackage{comment}
\usepackage{xcolor}

\usepackage{amsthm}

\newtheorem{thm}{Theorem}
\newtheorem{lem}{Lemma}

\theoremstyle{definition}
\newtheorem{defn}{Definition}

\theoremstyle{remark}
\newtheorem{rem}{Remark}
\newtheorem{ex}{Example}

\numberwithin{equation}{section}

\newcommand{\mc}{\mathcal}
\newcommand{\pd}{\,\partial}

\newcommand{\di}{\,\mathrm{d}}

\begin{document}

\title{Second Noether theorem for quasi-Noether systems}
\author{V. Rosenhaus and R. Shankar} 

\affil{\footnotesize{Department of Mathematics and Statistics, California State University, Chico, USA}}

\date{}

\maketitle

\begin{abstract}

\smallskip

Quasi-Noether differential systems are more general than variational systems and are quite common in mathematical physics. They include practically all differential systems of interest, at least those that have conservation laws.  In this paper, we discuss quasi-Noether systems that possess infinite-dimensional (infinite) symmetries involving arbitrary functions of independent variables. For quasi-Noether systems admitting infinite symmetries with arbitrary functions of all independent variables, we state and prove an extension of the Second Noether Theorem. 
In addition, we prove that infinite sets of conservation laws involving arbitrary functions of all independent variables are trivial and that the associated differential system is under-determined. We discuss infinite symmetries and infinite conservation laws of two important examples of non-variational quasi-Noether systems: the incompressible Euler equations and the Navier-Stokes equations in vorticity formulation, and we show that the infinite sets of conservation laws involving arbitrary functions of all independent variables are trivial.  We also analyze infinite symmetries involving arbitrary functions of not all independent variables, prove that the fluxes of conservation laws in these cases are total divergences on solutions, and demonstrate examples of this situation.

\end{abstract}

\noindent{\it Keywords}: second Noether theorem, infinite conservation laws, Navier-Stokes equations, vorticity equations, quasi-Noether systems, infinite symmetries

\section{Introduction}

Infinite-dimensional symmetry algebras parametrized by arbitrary functions and their 
relations to conservation laws have been studied considerably less extensively than finite-dimensional Lie 
symmetry groups and their corresponding conservation laws. For variational problems, according to the classical Noether result 
\cite{Noether} (see also \cite{Olver}), infinite variational symmetries with arbitrary functions of 
\emph {all independent variables} lead to certain identity relationships between the equations of the original 
system of partial differential equations (differential system) 
and their derivatives.  The most well known examples of such under-determined systems are gauge (invariant) theories.

The problem of a correspondence between variational infinite symmetries with arbitrary functions of \emph {not all independent 
variables} and conservation laws was studied in \cite{Rosenhaus02}. It was demonstrated that these infinite variational symmetries lead 
to a finite number of essential local conservation laws (those that give rise to non-vanishing conserved quantities) and that 
each essential conservation law is determined by a specific form of boundary condition. Using this method, conserved densities 
were found for a number of equations including the Zabolotskaya-Khokhlov equation \cite{Rosenhaus06a} and the Kadomtsev-Petviashvili 
equation \cite{Rosenhaus06b}. In \cite{Hydon}, Noether's approach was extended to include infinite variational symmetries involving functions of independent variables with some constraints.

A cohomological consideration of Noether's theorems in terms of
the variational bicomplex was given in 
\cite{Anderson,Navarro}. 
 
It was shown in \cite{Rosenhaus07} that the situation with infinite symmetry algebras parametrized by arbitrary functions of \emph{dependent variables} is radically different, leading to an infinite set of local conserved densities. Two known examples of this situation \cite{Rosenhaus13} are equations of Liouville type, see e.g. \cite{Zhiber2001}, that can be integrated by the Darboux method (e.g. \cite{Juras1997}), and hydrodynamic-type equations, \cite{Sheftel2004}.

\smallskip

\smallskip

However, Noether's approach is directly applicable only to the equations of a variational problem with a well-defined Lagrangian function. A main tool for our study is a method based on the Noether operator identity that establishes correspondences between symmetries and local conservation laws for differential systems without well-defined Lagrangian functions \cite{Rosenhaus94,Rosenhaus96}.  \emph{Quasi-Noether} systems are differential systems for which it is possible to establish a correspondence between symmetries and conservation laws based on the Noether operator identity, see \cite{Rosenhaus15}). Quasi-Noether differential systems are rather general and include all systems possessing conservation laws.

\smallskip

\smallskip

In the present paper, we apply an approach based on the Noether operator identity to quasi-Noether systems possessing infinite-dimensional symmetry algebras involving arbitrary functions of all independent variables. We show that, similarly to the Second Noether Theorem of variational systems, the existence of infinite symmetries with arbitrary functions of all independent variables results in differential identities relating equations of the original quasi-Noether system to their derivatives (showing that the original system is under-determined). We consider two important examples of such quasi-Noether systems, the incompressible Euler equations and the Navier-Stokes equations for vorticities, which are known to admit infinite sets of conservation laws.  We show that these conservation laws are trivial.  
\smallskip

\smallskip

We also analyze infinite symmetries involving arbitrary functions of not all independent variables. For these cases, we prove that the fluxes of the corresponding conservation laws are total divergences on solutions and demonstrate examples of this situation.

\smallskip

\smallskip
Section 2 discusses symmetries and conservation laws for variational systems, as well as the Noether identity and the two Noether theorems.
Section 3 introduces quasi-Noether systems and discusses the correspondence between symmetries and conservation laws for such systems. We formulate and prove extensions of the two Noether theorems for quasi-Noether systems. Section 4 deals with the case when conservation laws of a general differential system involve arbitrary functions of all independent variables. We show that such conservation laws are trivial and lead to differential identities in extended space. Section 5 discusses a case of infinite conservation laws that involve arbitrary functions of \emph{not all} independent variables. We demonstrate that the conservation laws here are, in general, not trivial and that their fluxes are total divergences on solutions. In Section 6, we consider two interesting examples of quasi-Noether systems: the incompressible Euler equations and the Navier-Stokes equations in vorticity formulation. We analyze infinite sets of conservation laws for these equations and show that they are trivial.

\section{Symmetries and conservation laws. Variational systems}

Let us briefly outline the approach we follow; for details, see \cite{Rosenhaus02,Rosenhaus07}.  

\begin{defn}[Conservation laws]
By a conservation law of a differential system 
\begin{align} \label{sys1}
\Delta_a(x,u,u_{(1)},u_{(2)},\dots, u_{(l)} ) = 0, \qquad a = 1,\dots,n, 
\end{align}
we mean a divergence expression
\begin{align} \label{CL}
\mathrm D_i K^i(x,u,u_{(1)},u_{(2)},\dots ) \doteq 0,
\end{align}
that vanishes on all solutions of the original system; we denote this type of equality by ($\doteq$). Here, $x=(x^1,x^2,\dots,x^p)$ and $u=(u^1,u^2,\dots,u^m)$ are the tuples of 
independent and dependent variables, respectively; 
$u_{(r)}$ is the tuple of $r$th-order derivatives of $u$,
and $r=1,2,\dots, l$. 
Each $K^i$ 
is a differential function (\cite{Olver}), i.e. a smooth function of~$x$, $u$, and a finite number of derivatives of~$u$. Each $\Delta_a,$ $a=1,2,...,n$ is a differential function of $x$, $u$, and all partial derivatives of each $u^v$, ($v=1,\dots,m$) with respect to the $x^i$ ($i=1,\ldots,p$) up to the $\ell$th order (\cite{Olver}).  Smooth functions $u^q=u^q(x),$ $\:q =1,\dots,m$ are defined on a non-empty open subset $D$  of $p$-dimensional space-time $\mathbb{R}^p$.
Summation over repeated indices is assumed.\qed
\end{defn}

\noindent We use  
the following notations for the differential system:

$(\Delta_1,\Delta_2,...,\Delta_n)\equiv \Delta \equiv \Delta[u] \equiv \Delta(x,u,u_{(1)},u_{(2)},\dots, u_{(l)} )$. 
In 
the
multi-index notation of \cite{Olver}, $x_J=(x_{j_1},x_{j_2},...,x_{j_k})$, 
and 
$u^v_J$ are partial derivatives, where 
$J=(j_1,j_2,...,j_k)$.

\begin{defn}[Trivial conservation laws]
A conservation law $\mathrm D_i P^i \doteq 0 $ is  \emph{trivial} \cite{Olver} if a linear combination of 
\emph{two kinds of  triviality} is taking place: 1. The $p$-tuple $P $ vanishes on the solutions of the original system: $P^i \doteq 0$. 2. The divergence identity is satisfied  for any point $[u]=(x,u_{(p)})$  in the 
extended jet space (e.g. $\pd_t\,\text{div}\,u+\text{div}(-\pd_t u)=0$).  Two conservation laws $K$ and $\tilde K$ are equivalent if they differ by a trivial conservation law.\qed
\end{defn}

\begin{rem}
If a conservation law \eqref{CL} exists, the following equality holds:
\begin{align} \label{CL charac}
D_iJ^i=\xi^a\Delta_a, 
\end{align}
where each $\xi^a$ is a \emph{characteristic}, and $J^i\,\dot{=}\,K^i$.  If $\xi^a= 0$, the conservation law \eqref{CL charac} is trivial; if $\xi^a\not\equiv 0$, it is nontrivial.\qed
\end{rem}

\begin{defn}[Euler operator]\label{defn:euler}
\begin{equation}\label{Rosenhaus:equation3}	
E_{u^a} = \frac{\partial }{\partial u^a} - \sum\limits_i {\mathrm D_i \frac{\partial }{\partial u_i^a } }
+ \sum\limits_{i \leqslant j} {\mathrm D_i } \mathrm D_j \frac{\partial }{\partial u_{ij}^a } + \cdots 
\end{equation}
is the Euler (Euler--Lagrange) operator (variational derivative).  In the notation of \cite{Olver}, we could give the Euler operator the following form:
\begin{align}\label{e}
E_{u^a}=(-\mathrm D)_J\frac{\partial}{\partial{u^a_J}},\hspace{4 mm}1\le a\le n,
\end{align}
the sum extending over all unordered multi-indices $J=(j_1,j_2,...,j_k)$ for
$1\le k\le p$. For $k=0$ we set $u^a_J=u^a$ and $\mathrm D_J=1$. Here and in what follows, we make use of the multi-index notation \cite{Olver} for total derivatives: $D_J=D_1^{J_1}\dots D_p^{J_p}$, where $J=(J_1,...,J_p)$ is a multi-index, $D_i=\partial_{x_i}+u^q_i\partial_{u^q}+\dots$ is a total derivative, and each $J_i$ is non-negative. The adjoint operator to $D_J$ is $(-D)_J=(-D_1)^{J_1}\dots(-D_p)^{J_p}$.
\end{defn}

\begin{rem}
The operator $E$ annihilates total divergences \cite{Olver}.  That is, $E_{u^q}(D_iJ^i)\equiv 0$.
\end{rem}

\begin{defn}[Variational problems]
Let
\[
S ={\int_D} {L(x,u, u_{(1)}, \dots)\: d^{p}x} 
\]											
be the action functional, where $L$ is the Lagrangian density. The equations of motion are
\begin{equation}\label{Rosenhaus:equation2}	
E_{u^a} (L) \equiv \Delta_a(x,u, u_{(1)}, \dots) = 0,
\end{equation}										
where $E$ is as in \eqref{Rosenhaus:equation3}.\qed
\end{defn}
 
\begin{defn}[Variational symmetry]
Consider a
one-parameter transformation with a canonical (vertical) infinitesimal operator 
\begin{align}\label{Rosenhaus:equation4}	
X_\alpha = \alpha ^a &\frac{\partial }{\partial u^a} + \sum\limits_i {(\mathrm D_i 
\alpha ^a)\frac{\partial }{\partial u_i^a }} + \sum\limits_{i \leqslant j} {(\mathrm D_i } 
\mathrm D_j \alpha ^a)\frac{\partial }{\partial u_{ij}^a } + \cdots  , \\
\quad \quad \alpha ^a & = \alpha ^a (x,u, u_{(1)}, \dots). \nonumber
\end{align}										
The variation of the functional $S$ under the transformation with operator 
$X_\alpha $ is
\begin{equation}								\label{Rosenhaus:equation5}	
\delta S = \int_D{X_\alpha L  \,d^{p}x}\,.
\end{equation}										
$X_\alpha $ is a variational (Noether) symmetry if 
\begin{equation}								\label{Rosenhaus:equation6}	
X_\alpha L =  \mathrm D_i M^i,
\end{equation}										
where $M^i=M^i (x,u,u_{(1)},\dots)$ are smooth functions of their arguments.\qed
\end{defn}

\begin{rem}
We can write the symmetry operator \eqref{Rosenhaus:equation4} in the form
\begin{align}
X_{\alpha} = \mathrm (D_{J}{\alpha}^a) \partial_{u^a_J},
\end{align}
where the sum is taken over all unordered multi-indices $J$.\qed
\end{rem}

The Noether identity \cite{Rosen} (see also \cite{Ibragimov1985} or \cite{Rosenhaus94} for the version used here) relates a symmetry operator 
$X_\alpha$ to the Euler- Lagrange operator $E_{u^q}$.
\begin{lem} [Noether Identity]
\label{lem:rosen}
Let $X_\alpha=(D_J\alpha^q)\pd/\pd u^q_J$ be a vertical symmetry operator with infinitesimal $\alpha$.  
The following operator identity holds:
\begin{align}\label{nid}
X_\alpha = {\alpha}^qE_{u^q}+\mathrm D_iR^i_\alpha,
\end{align}
where
\begin{align}\label{Ri}
R^{i}_\alpha = \alpha ^q\frac{\partial }{\partial u_i^q } + \left\{ 
{\sum\limits_{k \geqslant i} {\left( {\mathrm D_k \alpha ^q} \right) - \alpha 
^q\sum\limits_{k \leqslant i} {\mathrm D_k } } } \right\}\frac{\partial }{\partial 
u_{ik}^q }+ \cdots.
\end{align}\qed
\end{lem}

\begin{rem}
The operator $R^i_\alpha$ can be given a general form using multi-index notation \cite{Rosenhaus15}.
\begin{align}\label{Rii}
\begin{split}
R^i_\alpha&=\sum_{J}\sum_{m=1}^k\left[\delta^{ij_m}(-1)^{k+m}\left(\mathrm D_{J_{m-1}}{\alpha}^a\right)\mathrm D_{J-J_m}\right]\partial_{u^a_J},\\
&J_m=(j_1,j_2,...,j_m), \:\: J-J_m=(j_{m+1},j_{m+2},...,j_k), \:\: \mathrm D_{J_0}=\mathrm D_0=1.
\end{split}
\end{align}\qed
\end{rem}

\smallskip
Applying the identity \eqref{nid} to $L$ and using \eqref{Rosenhaus:equation6},	  
we obtain
\begin{equation}								\label{Rosenhaus:equation9}	
\mathrm D_i (M^i - R^i L) = \alpha ^a\Delta ^a, 
\end{equation}
leading to the statement of the First Noether Theorem: any 
one-parameter variational symmetry transformation with infinitesimal 
operator $X_\alpha$ \eqref{Rosenhaus:equation4} gives 
rise to a conservation law 										 
\begin{equation}								\label{Rosenhaus:equation10}	
\mathrm D_i (M^i - R^i_\alpha L)\doteq 0  
\end{equation}
(on the solution manifold $\Delta = 0$, $\mathrm D_i \Delta = 0$, \dots). 							

\begin{thm}[Noether 1]\label{thm:noether1}
Let $X_\alpha=(D_J\alpha^q)\pd/\pd u^q_J$ generate a finite variational symmetry group of a Lagrangian $\mc L$, such that $X_\alpha\mc L=D_iM^i_\alpha$ for some $M_\alpha$.  Then there exists the
following conservation law associated to the symmetry $X_{\alpha}$:
\begin{align}\label{c:noether1}
D_i(M^i_\alpha-R^i_\alpha\mc L)=\alpha^qE_{u^q}\mc L.
\end{align}
\end{thm}

\begin{rem}
Equation \eqref{c:noether1} is a conservation law since the right hand side is zero \emph{on-shell}, or when $E_{u^q}\mc L=0$ for each $q$.  In general, we say that computations are taken \emph{on-shell} if $u$ solves our differential system ($\Delta=0$).  Otherwise, we say that we work \emph{off-shell}.
\end{rem}

\begin{rem}
The first Noether theorem was formulated for continuous \emph{finite groups} of symmetry transformations.  For \emph{infinite groups} depending on arbitrary functions of all independent variables, we can formally generate the same type of 
continuity equations (conservation laws).  However, these equations will not have the same meaning as their finite-dimensional counterparts and will lead to \emph{differential identities} relating the equations of motion to their derivatives.
\end{rem}

\begin{thm}[Second Noether Theorem]\label{thm:noether2}
Let $X_{\alpha_f}$ \eqref{Rosenhaus:equation4} be a variational symmetry of Lagrangian $\mc L$ with infinitesimal $\alpha_f^q=\alpha^{qJ}D_Jf$ involving an arbitrary function of all independent variables $f=f(x_1,...,x_n)$.  Then $X_{\alpha_f}$ generates an identity involving the original equations and their derivatives. This off-shell identity takes the following form:
\begin{align}\label{id:noether}
(-D)_J\left[\alpha^{qJ}\Delta_q\right]=0.
\end{align}
\end{thm}

\begin{rem}
The identity \eqref{id:noether} expresses the fact that the original system $\Delta=E(\mc L)$ is under-determined.
\end{rem}

We next consider differential systems without well-defined Lagrangian functions.

\section{Symmetries and conservation laws of quasi-Noether systems}
\label{sec:quasi}

For a general differential system, a relationship between symmetries and 
conservation laws is unknown. In \cite{Rosenhaus94} and \cite{Rosenhaus96}, 
an approach
based on the Noether Identity \eqref{nid} 
was suggested to relate symmetries to conservation laws for a large class of differential systems that may not 
have well-defined Lagrangian functions. 

Consider a system \eqref{sys1}
\begin{align} \label{sys}
\Delta_a(x,u,u_{(1)},u_{(2)},\dots, u_{(l)}) = 0, \qquad a = 1,\dots,n,
\end{align}
of $n$ $\ell$th order differential equations for functions $u$.

Applying the Noether identity \eqref{nid}
to a combination of $\Delta_a$, we obtain
\begin{equation}								\label{Rosenhaus:equation12}	
X_\alpha ({\beta^a \Delta_a})= \alpha^v E_{u^v} (\beta^a \Delta_a) +
{D_{i} R^i(\beta^a \Delta_a)},
\quad a =1, \dots n, \:\: v=1, \dots, m, \:\: i=1, \dots, p.
\end{equation}
If there exist coefficients 
$\beta^a$ such that \cite{Rosenhaus96}
\begin{equation}\label{quasi-Noether}	
E_v (\beta^a\Delta_a) \doteq 0, \qquad a = 1, \dots n, \quad v=1, \dots, m,
\end{equation}
then each symmetry $X_\alpha$ of the system will lead to a local conservation law 
\begin{equation}								\label{Rosenhaus:equation16}	
\mathrm D_i R^i (\beta^a \Delta_a)\doteq 0, 
\end{equation}
for any differential system of class \eqref{quasi-Noether}. In \cite{Rosenhaus94}, 
the quantity $\beta^a\Delta_a$ was referred to as an alternative Lagrangian. 

Note that 
a condition similar to \eqref{quasi-Noether} (for a system to 
be quasi-Noether and possess a correspondence between symmetries and conservation laws; \cite{Rosenhaus94,Rosenhaus96}) was obtained earlier within an alternative approach based on the Lagrange identity; see 
\cite{Vladimirov1980,Vinogradov,Zharinov1986,Caviglia1986,Sarlet1987,Lunev}.
It was shown in \cite{Rosenhaus15} that, for a given symmetry transformation, both approaches give rise to the same conservation law. 

Note also that the same condition \eqref{quasi-Noether} has played a key role in the later developed 
nonlinear self-adjointness approach \cite{Ibragimov2011} to a correspondence between symmetries and conservation laws. Using this condition and alternative Lagrangians (that were called "formal Lagrangians"; see also \cite{Ibragimov07}), numerous applications and examples were given within this approach.  Note in addition that the idea of alternative Lagrangians was suggested in \cite{Olver} as Lagrangians in auxiliary variables.

\begin{defn}[Quasi-Noether systems]
We say that a system \eqref{sys1} $\Delta$ is \emph{quasi-Noether} if it has an \emph{alternative-Lagrangian} in the form of a linear combination of $\Delta_a$ that has zero variational derivative on-shell \cite{Rosenhaus15}.  
\smallskip
According to \eqref{quasi-Noether}, $\mc{A}=\beta^a\Delta_a$ is an alternative-Lagrangian if
\begin{align}\label{quasi}
E_{u^q}\mc A=\Xi^{qaJ}D_J\Delta_a,\quad 1\le q\le m,
\end{align}
where the $\beta$'s and $\Xi's$ are some functions. Clearly, $E_{u^q}\mc A=0$ when $\Delta=0$.
\end{defn}
For a regular Lagrangian of a variational problem, we can take $\mc A=\mc{L}$, such that \eqref{quasi} reduces to $E_{u^q}\mc L=\Delta_q$. 
\begin{rem}
Note that not all variational systems are quasi-Noether.  However, those variational systems that are not quasi-Noether are unlikely to be of physical interest, since such systems necessarily lack conservation laws, and, hence, lack continuous symmetries.  For example, the sine-Gordon-type Lagrangian
\begin{align*}
\mc L=-u_tu_x/2+f(t,x)\cos u
\end{align*}
leads to the following Euler-Lagrange equation
\begin{align*}
E(\mc L)=u_{tx}-f(t,x)\sin u=0.
\end{align*}
This equation admits no local conservation laws if $f(t,x)$ does not satisfy the following linear equation
\begin{align*}
\pd_t(a(t)f)+\pd_x(b(x)f)=0
\end{align*}
for any choice of $a(t)$ and $b(x)$ not both zero.  For example, $f(t,x)=\exp(-t^2x^2)/[1+(t+x)^2]$.
\end{rem}

\begin{rem}
Quasi-Noether systems are quite common and include most known variational systems and all differential systems in the form $D_iN^i_a=0$, such as the Euler and the Navier-Stokes equations; see \cite{Rosenhaus94} for other examples.  Most, if not all, interesting differential systems are quasi-Noether.
\end{rem}

\begin{thm}[Noether's First theorem for quasi-Noether systems]\label{thm:1}
Suppose that a quasi-Noether system \eqref{sys1} $\Delta$ possesses a finite continuous symmetry $X_\alpha$.  Then there exists conservation law \eqref{c:noe1} of the system $\Delta$ associated to the symmetry $X_\alpha$.
\end{thm}
\begin{proof}
Applying $X_\alpha$ to an alternative-Lagrangian $\mc A=\beta^a\Delta_a$ and using \eqref{Rosenhaus:equation12}, we obtain:
\begin{align}\label{x1}
\begin{split}
X_\alpha\mc A&=\alpha^qE_{u^q}\mc A  +D_iR^i_\alpha\mc A\\
&=\alpha^q\Xi^{qaJ}D_J\Delta_a+D_iR^i_\alpha\mc A.
\end{split}
\end{align}
Next, we rewrite the left hand side using the fact that $X_\alpha$ is a symmetry of $\Delta$, or that $X_\alpha\Delta=0$ on-shell.  Since $X_\alpha\Delta_b=\tau^{baJ}_\alpha D_J\Delta_a$ for some coefficients $\tau_\alpha$, we have:
\begin{align}\label{x2}
\begin{split}
X_\alpha\mc A&=X_\alpha(\beta^a\Delta_a)\\
&=\beta^aX_\alpha\Delta_a+(X_\alpha\beta^a)\Delta_a\\
&=(\beta^b\tau^{baJ}_\alpha D_J+X_\alpha\beta^a)\Delta_a\\
&=:\gamma^{aJ}_\alpha D_J\Delta_a.
\end{split}
\end{align}
Combining \eqref{x1} with \eqref{x2} yields the following:
\begin{align*}
\alpha^q\Xi^{qaJ}D_J\Delta_a+D_iR^i_\alpha\mc A=\gamma^{aJ}_\alpha D_J\Delta_a.
\end{align*}
Rearranging, we obtain a conservation law for each symmetry $\alpha$, (\cite{Rosenhaus94}):
\begin{align}\label{c:noe1}
D_iR^i_\alpha\mc A=(\gamma^{aJ}_\alpha-\alpha^q\Xi^{qaJ})D_J\Delta_a.
\end{align}
Indeed, the left hand side is a total divergence, and the right hand side is zero on-shell.
\end{proof}
\begin{rem}
The \emph{characteristic} of the conservation law \eqref{c:noe1} corresponding to the symmetry $X_{\alpha}$ of a quasi-Noether system 
$\Delta$ is as follows:
\begin{align} \label{charac}
\xi^a_\alpha =(-D)_J(\gamma^{aJ}_\alpha-\alpha^q\Xi^{qaJ}).
\end{align}
\end{rem}

\begin{ex}
Just to demonstrate an application of Theorem \ref{thm:1}, consider the Korteweg de-Vries (KdV) equation:
\begin{align}\label{e:kdv}
\Delta=u_t+uu_x+u_{xxx}=0.
\end{align}
Since this equation is, itself, a conservation law, the KdV equation is quasi-Noether: $\mc A=\Delta$, $E(\Delta)=0$.  The KdV equation  \eqref{e:kdv} is known to admit a scaling symmetry (see e.g. \cite{Olver}):
\begin{align}\label{s:kdv}
X_{\alpha}=(2u+xu_x+3tu_t)\pd_u.
\end{align}
The form of the conservation law resulting from $X_\alpha\mc A$ is immediate after using the commutator identity $[X_\alpha,D_t]=[X_\alpha,D_i]=0$:
\begin{align*}
X_\alpha\mc A&=D_t\alpha+D_x(\alpha u)+D_x^3\alpha.
\end{align*}
Expanding this expression, we can find that its characteristic is $\xi=1$. We obtain the following 
conservation law:
\begin{align}\label{c:kdv}
\begin{split}
D_t[2u+3tu_t+xu_x]&+D_x[u(2u+3tu_t+xu_x)]+D_x^3[2u+3tu_t+xu_x]\\&=
\Delta+D_t(3t\Delta)+D_x(x\Delta).
\end{split}
\end{align}
See \cite{Rosenhaus94} for more examples and details of applying Theorem \ref{thm:1}.
\end{ex}

We now formulate Noether's second theorem for quasi-Noether systems.

\begin{thm}[Noether's Second theorem for quasi-Noether systems]\label{thm:2}
Suppose that $X_{\alpha_f}$ is an infinite symmetry of quasi-Noether system 
\eqref{sys1} $\Delta$, where the infinitesimals $\alpha^q_f=\alpha^{qJ}D_Jf$ depend on an arbitrary function $f(x_1,...,x_p)$ of all independent variables.  Then $X_{\alpha_f}$ generates identity \eqref{id} involving the equations of the system $\Delta$ and their derivatives. 
\end{thm}
\begin{proof}
For infinite symmetry group operator $X_{\alpha_f}$, we can apply the same procedure as in Theorem \ref{thm:1} to formally obtain analogous conservation laws
\eqref{c:noe1} with characteristic \eqref{charac}; here, however, the infinite symmetry $X_{\alpha_f}$ leads to an infinite set of 
conservation laws
involving the arbitrary function $f$ and its derivatives:
\begin{align}\label{c:inf}
D_iJ^i_{\alpha_f}=(\xi^{aJ}D_J f) \Delta_a
\end{align}
for some functions $\xi^{aJ}$. We follow Noether's approach and show that these conservation laws have a different meaning; they lead to differential identities involving the equations of the
system $\Delta$ and their derivatives and, hence, express the under-determinacy of the system $\Delta$. Our computations are performed \emph{off-shell}.  Integrating \eqref{c:inf} over an arbitrary bounded connected subset $S$ of $\mathbb{R}^p$, we obtain an integral identity:
\begin{align*}
\int_SD_iJ^i_{\alpha_f}\di^px=\int_S(\xi^{aJ}D_Jf)\,\Delta_a\di^px.
\end{align*}
We integrate the right hand side by parts using the identity \eqref{ibp} with \eqref{Phi} and collect the divergences to the left hand side:
\begin{align*}
\int_SD_i\left(J^i_{\alpha_f}-\Phi^i_J[f,\xi^{aJ}\Delta_a]\right)\di^px=\int_ S f(-D)_J\left[\xi^{aJ}\Delta_a\right]\di^px.
\end{align*}
We observe that, inside the divergence operator, there are two vectors: $J^i_{\alpha_f}$ and $\Phi^i_J[f,\xi^{aJ}\Delta_a]$. The first one is:
\begin{align}
J^i_{\alpha_f} = R^i_{\alpha_f} \mc A-\Phi^i_J[\Delta_a,\gamma^{aJ}_{\alpha_f}-\alpha^q_f\Xi^{qaJ}],
\end{align}
where $\mc A=\beta^a\Delta_a$, $R^i_{\alpha f}$ is expressed through $\alpha_f$ and its derivatives according to \eqref{Ri}, and $\Phi^i$, expressed in \eqref{Phi}, comes from integrating the right hand side (RHS) 
of \eqref{c:noe1} by parts to obtain \eqref{c:inf}. Therefore, the vector $J^i_{\alpha_f}$ is a linear combination of the arbitrary function $f(x_1,...,x_p)$ and its derivatives. The same conclusion can be made with respect to the second vector $\Phi^i_J[f,\xi^{aJ}\Delta_a]$. We choose the arbitrary function of all independent variables $f(x)$ so that it vanishes on $\partial S$ together with all its derivatives. Then an application of Gauss's theorem to the left hand side reduces it to vanishing surface terms, and, therefore, we obtain:
\begin{align}\label{iden}
0=\int_Sf(-D)_J\left[\xi^{aJ}\Delta_a\right]\di^px.
\end{align}
The equation \eqref{iden} holds for arbitrary functions $f$ that vanish on $\partial S$ together with their derivatives, and, therefore, it is an identity. Since the bounded subset $S$ is arbitrary, the integrand must vanish identically everywhere in the extended space:
\begin{align}\label{id}
(-D)_J\left[\xi^{aJ}\Delta_a\right]=0.
\end{align}

The identity \eqref{id} is a functional relationship between the original equations of the system $\Delta$ and their derivatives.  Thus, not all equations of  $\Delta$ are independent, and the system $\Delta$ is \emph{under-determined}.  The existence of an infinite symmetry $X_{\alpha_f}$ and the infinite conservation laws \eqref{c:inf} is a consequence of the additional degree(s) of freedom related to the description of an under-determined system $\Delta$.
\end{proof}

\begin{rem}
We express the identity \eqref{id} in terms of the symmetry coefficients $\alpha^q=\alpha^{qJ}D_Jf$.  In \eqref{x2}, we suppose that $X_\alpha\beta^a=\overline\beta^{aJ}D_Jf$, and $X_\alpha\Delta_b=\overline\tau^{baJK}D_KfD_J\Delta_a$.  Then the characteristic \eqref{charac} in \eqref{c:inf} is as follows:
\begin{align}
\xi^a&=(-D)_J\left[\left(\overline\beta^{aK}\delta^{J0}+\beta^b\overline\tau^{baJK}-\alpha^{qK}\Xi^{qaJ}\right)D_Kf\right] \nonumber \\
&=(-D)_J\left[\overline\beta^{a0}\delta^{J0}+\beta^b\overline\tau^{baJ0}-\alpha^{q0}\Xi^{qaJ}\right]f+\dots\\
&=:\xi^{a0}f+\dots  \nonumber
\end{align}
\end{rem}

\begin{rem}
In terms of these coefficients, we can rewrite identity \eqref{id} as follows:
\begin{align}\label{NIIcalc}
\xi^{a}\Delta_a-D_i(\xi^{ai}\Delta_a)+\sum_{i\le j}D_iD_j(\xi^{aij}\Delta_a)-\sum_{i\le j\le k}D_iD_jD_k(\xi^{aijk}\Delta_a)+\dots=0.
\end{align}
\end{rem}

\begin{rem}
Alternative Lagrangians do not correspond to well-defined variational problems, and true variational formulas are not the same as the ones obtained with an alternative Lagrangian of quasi-Noether systems.  For example, we can see that \eqref{id:noether} is different from \eqref{id} since the functions $\xi^{aJ}$ are, clearly, different from the functions $\alpha^{qJ}$; see \eqref{charac}.
\end{rem}

\section{Arbitrary functions, infinite conservation laws, and differential identities}

\subsection{Overview}
In this section, we discuss a general case of conservation laws that contain arbitrary functions of \emph{all} independent variables. We show that these conservation laws are necessarily trivial. To illustrate our conclusion, we consider the special case of a first order differential equation $\Delta(t,x,u,u_t,u_x)=0$ for a single function $u$ of two variables $t$ and $x$.  Suppose that $\Delta$ admits the following infinite set of conservation laws:
\begin{align}\label{c:ex}
D_tM^t_f+D_xM^x_f=(f\xi^0+\pd_tf\,\xi^{10}+\pd_xf\xi^{01})\Delta,
\end{align}
where $f(t,x)$ is an arbitrary function of \emph{all independent variables}, $\xi^i$ are some functions of $(t,x,u)$, and the coefficients $M^i_f$ involve $f$ and its derivatives.
Rewriting the right hand side, we obtain
\begin{align}\label{c:exx}
D_tM^t_f+D_xM^x_f= \xi^0 \Delta f-D_t(\xi^{10}\Delta) f+ D_t(\xi^{10}\Delta f)-D_x(\xi^{01}\Delta)f + D_x(\xi^{01}\Delta f).
\end{align}
Equivalently,
\begin{align}\label{c:ex:3}
D_t(M^t_f-f\xi^{10}\Delta)+D_x(M^x_f-f\xi^{01}\Delta)=f[\xi^0\Delta-D_t(\xi^{10}\Delta)-D_x(\xi^{01}\Delta)].
\end{align}
We integrate the left hand side over the whole space $t,x$. Similarly to derivation of \eqref{id}, we can choose the arbitrary function of all independent variables $f(t,x)$ to vanish at the boundaries together with all its derivatives. Then, using Gauss' theorem and the fact that the surface terms vanish, we obtain
\begin{align}\label{id:ex}
\int_S f [(\xi^0\Delta-D_t(\xi^{10}\Delta)-D_x(\xi^{01}\Delta)] dx dt=0.
\end{align}
Equation \eqref{id:ex} holds for an arbitrary function $f$ that vanishes on $\partial S$ together with its derivatives, so the integrand must vanish identically on $S$:
\begin{align}\label{id:exx}
\xi^0\Delta-D_t(\xi^{10}\Delta)-D_x(\xi^{01}\Delta)=0.
\end{align}
The identity \eqref{id:exx} is a functional relationship between $\Delta$ and its derivatives that holds for all smooth functions $u(t,x)$.  Thus, the equation $\Delta$ for $u=u(t,x)$ is \emph{under-determined}.

The final step to prove the triviality of \eqref{c:ex} (\eqref{c:ex:3}) is to substitute the identity 
\eqref{id:exx} into \eqref{c:ex:3}. We rewrite \eqref{c:ex:3}:
\begin{align}\label{c:ex:4}
D_t(M^t_f-f\xi^{10}\Delta)+D_x(M^x_f-f\xi^{01}\Delta)=0.
\end{align}
Equation \eqref{c:ex:4} is an identity that holds for all functions $u=u(t,x)$, not only solutions of $\Delta=0$. Since the coefficient of $\Delta$ on the right hand side is zero, \emph{this conservation law is trivial}.  

\subsection{Arbitrary functions of all independent variables, and differential identities}
For general differential systems, we have the following result.
\begin{thm}\label{thm:triv}
Consider a differential system 
\eqref{sys1} $\Delta \equiv (\Delta_1,\Delta_2,...,\Delta_n)$ for functions $u^q,1\le q\le m$.  
Suppose that $\Delta$ possesses an infinite conservation law
\begin{align}\label{c:inf1}
D_iM^i_f=\xi^a_f\Delta_a,
\end{align}
where the fluxes $M^i_f$ and characteristics $\xi^a_f$ are given by
\begin{align}\label{malpha}
\begin{split}
M^i_f=M^{iJ}D_Jf,\quad 1\le i\le p,\\
\xi^a_f=\xi^{aJ}D_Jf,\quad 1\le a\le n,
\end{split}
\end{align}
and $f(x)=f(x_1,...,x_p)$ is an arbitrary function of all independent variables.

Then there exists differential identity \eqref{id1} involving the equations of the system $\Delta$ and their derivatives (i.e. $\Delta$ is an under-determined system), and the infinite conservation law \eqref{c:inf1} is trivial.
\end{thm}

We present an algebraic proof of this theorem using some tools implemented in \cite{Olver}.  
We recall the Euler operator $E_f$ in Definition \ref{defn:euler} and the following integration by parts identity.

\begin{lem}
Let $f$ and $g$ be functions, $J=(J_1,J_2,...,J_p)$ be a multi-index, and $D_J=D_1^{J_1}\dots D_p^{J_p}$ be a total derivative.  Then the following identity holds:
\begin{align}\label{ibp}
g\,D_Jf=f\,(-D)_Jg+D_i\Phi^i_J[f,g],
\end{align}
where $(-D)_J=(-D_1)^{J_1}\dots(-D_p)^{J_p}$ is the adjoint operator to $D_J$, and 
\begin{align} \label{Phi}
\Phi^i_J[f,g]=\sum_{j=0}^{J_i-1}\left(D_i^{J_i-1-j}D_{i+1}^{J_{i+1}}\dots D_n^{J_n}f\right)\,\,\left((-D_i)^j(-D_1)^{J_1}\dots(-D_{i-1})^{J_{i-1}}g\right),
\end{align}
for each $1\le i\le p$.
\end{lem}

\begin{ex}
\begin{align*}
g\,D_x^2D_yf=(-1)^3f\,D_x^2D_yg+D_x\left(D_xD_yf\,g-D_yf\,D_xg\right)+D_y(f\,D_x^2g).
\end{align*}
\end{ex}

The identity \eqref{ibp} helps us directly apply the Euler operator \eqref{e} instead of integrating over space, as in Noether's approach. 
\begin{proof}[Proof of Theorem \ref{thm:triv}]

Using the integration by parts identity \eqref{ibp} in the RHS of \eqref{c:inf1}, we obtain:
\begin{align}\label{c:inf:1}
D_iM^i_f=D_Jf\,\xi^{aJ}\Delta_a=f(-D)_J\left[\xi^{aJ}\Delta_a\right]+D_i\Phi^i_J[f,\xi^{aJ}\Delta_a].
\end{align}
Applying now the Euler operator $E_f$ to \eqref{c:inf:1}, we obtain:
\begin{align}\label{id1}
0=(-D)_J\left[\xi^{aJ}\Delta_a\right].
\end{align}
Equation \eqref{id1} is a differential identity holding for all functions $u$.  It relates $\Delta$ to its derivatives and shows that the system $\Delta$ is under-determined.

Substituting \eqref{id1} into \eqref{c:inf:1} gives an equivalent representation of \eqref{c:inf1}:
\begin{align}\label{c:inf:2}
D_iM^i_f=D_i\Phi^i_J[f,\xi^{aJ}\Delta_a].
\end{align}
The right hand side is a total divergence.  Thus, the characteristic of this conservation law is zero, and the infinite set of conservation laws \eqref{c:inf1} is trivial.
\end{proof}
\begin{rem}
Any differential system that has conservation laws is necessarily quasi-Noether, according to Definition \ref{quasi-Noether}.  Therefore, Theorem \ref{thm:triv} is formulated for quasi-Noether systems.
\end{rem}
\medskip  
\begin{ex}
An example of Theorem \ref{thm:triv} is a known fact that (infinitesimal) local gauge transformations lead to trivial conservation laws (continuity equations).  Consider a scalar field $\varphi$ interacting with an electromagnetic field (scalar electrodynamics with zero potential) with the Lagrangian:
\begin{align}\label{l:sc}
\mathcal{L}=\mc{D}_\mu\bar\varphi\,\mc{D}^\mu\varphi-\frac{1}{4}F_{\mu\nu}F^{\mu\nu},
\end{align}
where $\bar\varphi$ is the complex conjugate of $\varphi$, $\mc D_\mu=\partial_\mu-iA_\mu$, and 
$F_{\mu\nu}=\partial_\mu A_\nu-\partial_\nu A_\mu$,   $\mu, \nu = 1, \ldots,4.$  

The Euler-Lagrange equations for \eqref{l:sc} are as follows:
\begin{align}\label{e:sc:1}
E_{A^\mu}\mathcal{L}&=i(\bar\varphi\,\pd_\mu\varphi-\varphi\pd_\mu\bar\varphi)+2A_\mu|\varphi|^2-\partial^\nu F_{\mu\nu}=0, \nonumber \\
E_{\bar\varphi}\mathcal{L}&=-\pd_\mu\pd^\mu\varphi+2iA^\mu\pd_\mu\varphi+i\varphi\pd_\mu A^\mu+2A_\mu A^\mu\varphi=0, \\
E_{\varphi}\mc L &=\overline{E_{\bar\varphi}\mc L}.     \nonumber
\end{align} 

Lagrangian \eqref{l:sc} is invariant with respect to the following infinite symmetry (local gauge) transformation:
\begin{align}\label{s:sc}
X_\theta=i\theta\varphi\pd_\varphi-i\theta\bar\varphi\pd_{\bar\varphi}+\pd^\mu\theta\pd_{A^\mu},
\end{align}
where $\theta=\theta(x)$ is an arbitrary function of all independent variables $x_\mu$, $X_\theta\mathcal{L}\equiv 0$.

Applying the Noether operator identity \eqref{nid} to the Lagrangian \eqref{l:sc}, we obtain the following infinite 
set of conservation laws \eqref{Rosenhaus:equation10}:
\begin{align}\label{c:sc}
\begin{split}
& \partial^\mu M_\mu\doteq 0, \\
& M_\mu =\theta\left[i\left(\bar\varphi\pd_\mu\varphi-\varphi\pd_\mu\bar\varphi\right)+2A_\mu|\varphi|^2\right]+\partial^\nu\theta F_{\mu\nu}.
\end{split}
\end{align}
Using equations \eqref{e:sc:1}, we can rewrite the fluxes $M_\mu$ in the following form:
\begin{align*}
M_\mu=\theta E_{A^\mu} \mathcal{L}+\partial^\nu(\theta F_{\mu\nu}).
\end{align*}
The first term vanishes on the solutions of equations \eqref{e:sc:1}, and, therefore, corresponds to a trivial conservation law of the first type. The second term corresponds to a trivial conservation law of the second type; its divergence vanishes identically since the tensor $F_{\mu\nu}$ is anti-symmetric. Therefore, according to the second part of Theorem \ref{thm:triv}, the infinite conservation law \eqref{c:sc} with an arbitrary function $\theta(x)$ of all independent variables $x_\mu$ is trivial. 

Consistent with the first part of Theorem \ref{thm:triv} (and the Second Noether Theorem), the following identity relationship \eqref{id1} between equations \eqref{e:sc:1} and their derivatives holds:
\begin{align}
i\varphi E_\varphi\mc L-i\bar\varphi E_{\bar\varphi}\mc L-\partial^\mu (E_{A^\mu}\mc L)\equiv 0.
\end{align} 
\end{ex}

Examples of infinite conservation laws involving arbitrary functions of all independent variables for non-Lagrangian systems are demonstrated in Section 6.
 
\section{Arbitrary functions of not all independent variables, and conservation laws}
In this section, we discuss the general case of infinite conservation laws that contain arbitrary functions of \emph{not all} independent variables. Theorem \ref{thm:triv} is valid for $f(x_1,...,x_p)$, or for arbitrary functions $f$ of \emph{all} independent variables.  It is not valid for arbitrary functions of \emph{not all} independent variables, e.g. $f(x_1)$. We present an extension of Theorem \ref{thm:triv} to this case.  

\begin{thm}\label{thm:div}
Suppose that the arbitrary function $f$ in Theorem \ref{thm:triv} depends only on the variables $(x_r,x_{r+1},...,x_p)$ for some $1< r\le p$.  Then there exists another conservation law of the system $\Delta:$ 
\eqref{id:2}, different from \eqref{c:inf1}, where the fluxes $M^r_f,M^{r+1}_f,...,M^p_f$ \eqref{malpha} are total divergences on-shell ($\Delta=0$).
\end{thm}

An application of the variational derivative gives rise, instead of to the differential identity \eqref{id1}, to \emph{another conservation law that does not involve $f$}. We also find that the conservation law \eqref{c:inf1} here is not trivial, and its \emph{fluxes are total divergences on solutions}.  In particular, a conservation law involving $f(t)$ has a density $M^t_f$ that is a spatial divergence on solutions.

\begin{proof}[Proof of Theorem \ref{thm:div}]
$\,$
\medskip

Part 1. According to \eqref{c:inf:1}:
\begin{align}\label{c:inf:3}
\sum_{i=1}^{r-1}D_iM^i_f+\sum_{i=r}^{p}D_iM^i_f=f(-D)_J[\xi^{aJ}\Delta_a]+\sum_{i=r}^p D_i\Phi^i_J[f,\xi^{aJ}\Delta_a],
\end{align}
where $J=(J_r,J_{r+1},...,J_p)$ is a multi-index. Applying the variational derivative $E_f$ \eqref{e} annihilates the second and fourth terms:
\begin{align}\label{id:1}
\sum_{i=1}^{r-1}E_fD_iM^i_f=(-D)_J[\xi^{aJ}\Delta_a].
\end{align}
Using \eqref{malpha} and the integration by parts identity \eqref{ibp}, we have
\begin{align}\label{mibp}
D_iM^i_f=D_JfD_iM^{iJ}=f\,(-D)_JD_iM^{iJ}+\sum_{j=r}^pD_j\Phi^j_J[f,D_iM^{iJ}],\quad 1\le i\le r-1.
\end{align}
Substituting \eqref{mibp} into \eqref{id:1}, we obtain:
\begin{align}\label{id:2}
\sum_{i=1}^{r-1}D_i\left[(-D)_JM^{iJ}\right]=(-D)_J[\xi^{aJ}\Delta_a].
\end{align}
Equation \eqref{id:2} is a conservation law distinct from \eqref{c:inf1}. It is nontrivial if $\xi^{a0}\neq 0$ for some $a$.
\smallskip

Part 2.  To prove the second part of the Theorem, we multiply \eqref{id:2} by $f$ and subtract it from \eqref{c:inf:3}, noting that $\Phi^j_J[f,D_iM^{iJ}]=D_i\Phi^j_J[f,M^{iJ}]$ in \eqref{mibp}:
\begin{align}\label{c:triv}
\sum_{i=r}^nD_i\sum_{j=1}^{r-1}D_j\Phi^i_J[f,M^{jJ}]+\sum_{i=p}^nD_iM^i_f=\sum_{i=p}^nD_i\Phi^i_J[f,\xi^{aJ}\Delta_a].
\end{align}
The general solution of this equation for the $M^i_f$ is as follows:
\begin{align}\label{sol:flux}
M^i_f=-\sum_{j=1}^{r-1}D_j\Phi^i_J[f,M^{jJ}]+\Phi^i_J[f,\xi^{aJ}\Delta_a],\quad r\le i\le n,
\end{align}
where it is understood that $M^i_f$ is arbitrary up to $M^i_f\to M^i_f +\lambda^i+\sum_{j=r}^nD_jR^{ij}_f$, where $\lambda^i$ is a constant, and each $R^{ij}_f=-R^{ji}_f$ is an arbitrary function (i.e. up to fluxes of trivial conservation laws).

From \eqref{ibp}, we see that each $\Phi^i_J[f,\xi^{aJ}\Delta]$ is a linear function of $\Delta$ and vanishes when $\Delta=0$.  Thus, on solutions $u$ of $\Delta=0$, each flux $M^i_f$ in \eqref{sol:flux} is a total divergence.
\end{proof}

\subsection{Examples of Theorem \ref{thm:div}}
Consider Liouville's equation for $u(t,x)$:
\begin{align}\label{e:liou}
\Delta=u_{tx}-e^u=0
\end{align}
with the Lagrangian:
\begin{align}\label{l:liou}
\mathcal{L}=-\frac{1}{2}u_tu_x-e^u.
\end{align}
This Lagrangian admits an infinite symmetry:
\begin{align}\label{s:liou}
X=[f'(t)+f(t)u_t]\partial_u=:\xi_f\pd_u
\end{align}
and an infinite set of corresponding conservation laws:
\begin{align}\label{c:liou}
\frac{1}{2}D_t\left[u_xf_t-2fe^u\right]+\frac{1}{2}D_x\left[(f_t+fu_t)u_t-f_{tt}u\right]=\xi_f E_u\mathcal{L}=\xi_f\Delta.
\end{align}
We write 
\begin{align*}
\xi_f\Delta=f(u_t\Delta-D_t\Delta)+D_t(f\Delta).
\end{align*}
Applying the Euler operator identity \eqref{id:1}, we obtain:
\begin{align}\label{kaput}
u_t\Delta-D_t\Delta= (-D)_J[\xi^{aJ}\Delta_a].
\end{align}
Since
\begin{align*}
\xi^{a}_f=\xi^{aJ}D_J f = u_t f(t) + f'(t),
\end{align*}
it follows that
\begin{align*}
\xi^{0}= u_t, \qquad \xi^{t}=1.
\end{align*}
Calculating the RHS of \eqref{kaput}, we obtain:
\begin{align}\label{kaputa}
(-D)_J[\xi^{aJ}\Delta_a]=u_t(u_{tx}-e^u) -D_t (u_{xt} - e^u) = u_t u_{xt} -u_{xtt} =D_x\left[\frac{1}{2}u_t^2-u_{tt}\right],
\end{align}
and 
\begin{align}\label{c:liou1}
u_t\Delta-D_t\Delta=D_x\left[\frac{1}{2}u_t^2-u_{tt}\right],
\end{align}
which can be easily verified.  This is the additional conservation law independent of $f$ predicted by Theorem \ref{thm:div}, since the left hand side is zero when $\Delta=0$.   In fact, this is a first integral of Liouville's equation.

We now demonstrate the second part of Theorem \ref{thm:div} and rewrite the density of \eqref{c:liou} as a total divergence on shell:
\begin{align*}
u_xf_t-2fe^u=u_xf_t+2f\Delta-2fu_{tx}= D_x(uf_t-2fu_t)+2f\Delta.
\end{align*}
Thus, when $\Delta=0$, we see that the density is a total derivative in $x$.
\medskip

Consider another example, \emph {the equation for non-stationary transonic gas flow}
\begin{align}  \label{transeq}
\Delta=2u_{xt}+u_xu_{xx}-u_{yy}
\end{align}
with the Lagrangian function
\begin{align}
L=-u_xu_t - \frac{u_x^3}{6} + \frac{u_y^2}{2}.
\end{align}
The conservation law in question is:
\begin{align}    \label{cont}	            
D_tM^t_f+D_xM^x_f+D_yM^y_f=\xi_f\Delta,
\end{align}
where 
\begin{align}     \label{alpha}
\xi_f=2xf'(t)+2y^2f''(t)-f(t)u_x,
\end{align}
and the density and fluxes are given by
\begin{align} \label{fluxes}
\begin{split}
M^t_f&=f(t)[-u_x^2]+f'(t)[2xu_x-2u]+f''(t)[2y^2u_x],\\
M^x_f&= - L f(t) - 2xuf''(t)-2y^2uf'''(t) + \xi_f(u_t+u_x^2/2), \\
M^y_f&= 4yuf''(t) - \xi_fu_y,
\end{split}
\end{align}
We find that the density can be rewritten as a total divergence on solutions of $\Delta=0$, in agreement with Theorem \ref{thm:div}:
\begin{align} \label{Mt}
M^t_f\equiv Q^t_{f\Delta}+D_xT^x_f+D_yT^y_f,
\end{align}
where
\begin{align}\label{tq}
\begin{split}
T^x_f&=f(t)[-4xu_t+4y^2u_{tt}+2y^2u_xu_{xt}-xu_x^2]+f'(t)[2xu-4y^2u_t-y^2u_x^2]+f''(t)[2y^2u],\\
T^y_f&=f(t)[2xu_y+4yu_t-2y^2u_{yt}]+f'(t)[2y^2u_y-4yu],\\
Q^t_{f\Delta}&=(2xf(t)+2y^2f'(t))\Delta-2y^2f(t)D_t\Delta.
\end{split}
\end{align}

Let us show how Theorem \ref{thm:div} can be used to generate conserved densities. Integrating the 
conservation law
\eqref{cont} over the whole space on shell, we will get
\begin{align} \label{intcont}                                
D_t\int M^t_f dx dy + D_x\int M^x_f dx dy + D_y\int M^y_f dx dy \: \doteq \; 0,
\end{align}
where $M^x_f, M^y_f$ are determined in \eqref{fluxes}, and $M^t_f$ is given by \eqref{Mt} and \eqref{tq}.
Equation \eqref{intcont} leads to a 
quantity conserved in time
\begin{align} \label{cons}                                
D_t\int M^t_f dx dy \: \doteq \; 0
\end{align}
if the contributions (surface terms) from the second and third 
integrals
vanish, i.e.
\begin{align}
\begin{split}
M^x {\Big|}_{x \to \pm \infty}\;\mathop\to \limits \: 0, \qquad
M^y {\Big|}_{y \to \pm \infty}\;\mathop\to \limits \: 0.
\end{split}
\end{align}
It is easy to see that, for sufficiently strict boundary conditions, all surface terms
vanish. But for these boundary conditions, all terms in $T^x$, and $T^y$ also vanish,
and we get no conserved densities of
the equation \eqref{transeq}.

However, for weaker boundary conditions, the situation changes. Consider a ``regular" asymptotic behavior 
\begin{align}
\begin{split} \label{regularbound}
u,u_i \,\, \mathop\to \limits_{x \to \pm \infty } \: & 0,\\
u,u_i \,\, \mathop\to \limits_{y \to \pm \infty } \: & 0.
\end{split}
\end{align}  \label{regularbc}
In this case, after eliminating vanishing terms, we obtain
\begin{align}
\begin{split}
M^x_f {\Big|}_{x \to \pm \infty}\;\mathop\to \limits \: & f'(t)2x(u_t+u_x^2/2) - 2xu f''(t), \\
M^y_f {\Big|}_{y \to \pm \infty}\;\mathop\to \limits \:& f''(t)(4yu - 2y^2u_y).
\end{split}
\end{align}   \label{MxMyregular}
Requiring these terms to vanish: 
\begin{align}
f(t)=a, \qquad a=\,\text{constant}.
\end{align}
$T^x_f$, and $T^y_f$ reduce to
\begin{align}
\begin{split}
T^x &= -4xu_t -xu_x^2 + 2y^2(2u_{tt} +u_xu_{xt}),  \\
T^y &= 4yu_t -2y^2u_{ty} +2x u_y.
\end{split}
\end{align}
Calculating $M^t_f$ for this case, taking into consideration equation \eqref{transeq}, and noting that the last term in $T^y$ vanishes when integrated over $y$, we obtain
\begin{align}
M^t_f \doteq T^x_{,x} +T^y_{,y} = - u_x^2-4xu_{xt}- 2xu_xu_{xx}.
\end {align}
The last two remaining terms in $M^t$ are:
\begin{align}
-2x(2u_{xt} + u_xu_{xx}) \doteq -2x u_{yy} = D_y(-2xu_y).
\end{align}
The contribution of this term in the integral over $y$ is, obviously, zero.
Finally, we have
\begin{align}
M^t_f \doteq - u_x^2.
\end{align}
Thus, for regular boundary conditions \eqref{regularbound}, we obtain the following conserved 
quantity
\begin{align}                                
D_t\int u_x^2 \, dx \, dy \doteq \, 0.
\end{align}
In \cite{Rosenhaus02}, this result was obtained using a different technique.

\section{Arbitrary functions of all independent variables in vorticity-type systems}
\subsection{General systems}
We consider a class of vorticity-type systems studied in \cite{Cheviakov} and \cite{Chevober}:
\begin{align}\label{e:ch}
\begin{split}
\vec\Delta&=\vec \omega_t+\nabla\times\vec M(t,x,\omega,...)=0,\quad\text{or}\quad\Delta_i=\omega^i_t+\epsilon^{ijk}D_jM^k=0,\quad i=1,2,3,\\
\Delta_4&=\nabla\cdot\vec \omega=0,\quad\text{or}\quad\Delta_4=\omega^i_i=0.
\end{split}
\end{align}
Here, $\omega$ is the vorticity, and $\vec M$ is some vector function.

The following infinite set of conservation laws of \eqref{e:ch} involving an arbitrary function $F(t,x)$ was presented in \cite{Cheviakov} and \cite{Chevober}:
\begin{align}\label{c:ch}
\pd_t(\vec\omega\cdot\nabla F)+\nabla\cdot\left(\vec M\times\nabla F-\pd_tF\,\vec\omega\right)=\nabla F\cdot\vec\Delta-\pd_tF\,\Delta_4=0.
\end{align}

We show that the conservation laws \eqref{c:ch} are trivial.  Indeed, applying the procedure of Theorem \ref{thm:triv} to \eqref{c:ch} (i.e. taking the variational derivative $E_F$ of \eqref{c:ch}) leads us to the following off-shell identity:
\begin{align}\label{id:ch}
D_t\Delta_4-\nabla\cdot\vec\Delta=0.
\end{align}
Substituting \eqref{id:ch} into \eqref{c:ch} yields:
\begin{align}\label{c:ch:1}
\pd_t(\vec\omega\cdot\nabla F)+\nabla\cdot\left(\vec M\times\nabla F-\pd_tF\,\vec\omega\right)=\nabla\cdot(F\vec\Delta)-D_t(F\Delta_4).
\end{align}
Since the characteristic (coefficient of $\Delta$) of this conservation law is zero, the conservation law \eqref{c:ch} is trivial.  It is a combination of two kinds of triviality: 
the divergence terms on the RHS have fluxes that vanish on solutions (on shell, $\Delta=0$),
and, therefore, do not contribute to the conserved densities or integrals. The rest is a differential identity that holds everywhere in the extended space (off-shell). 

The relation (\ref{id:ch}) expresses an interdependence of equations in the governing differential system $\Delta$, or the fact that the system is under-determined (abnormal, see \cite{Olver}).  The relation \eqref{id:ch} was noted in \cite{Chevober} in the case of the Navier-Stokes equations, but not the fact that the infinite conservation laws \eqref{c:ch} are trivial.

\subsection{Symmetries of the Euler vorticity system}
We now apply Theorem \ref{thm:2} to the vorticity system \eqref{e:ch} for velocity $u$ and vorticity $\omega$:
\begin{align}
\begin{split}\label{e:eu}
\Delta_i&=\omega^i_t+\partial_j(\omega^iu^j-\omega^ju^i)=0,\qquad i=1,2,3,\\
\Delta_4&=\omega^i_i=0,\\
\Delta_5&=u^i_i=0.
\end{split}
\end{align}
This system is clearly quasi-Noether, and we can set alternative Lagrangians as $\mc A=\Delta_i$ for each $i=1,2,3$.  It admits several infinite sets of symmetries.  If functions $g^k(t,x)$ are arbitrary, and $f^i=\epsilon^{ijk}\partial_jg^k$ ($\partial_i f^i=0$), then the infinite symmetries are given by:
\begin{align}\label{s:eu}
\begin{split}
X_f&=\,\alpha_f^i\pd_{\omega^i}+\beta_f^i\pd_{u^i},\\
\alpha^i_f&=\omega^j\pd_jf^i-f^j\pd_j\omega^i,\\
\beta^i_f&=\pd_tf^i+u^jf\pd_jf^i-f^j\pd_ju^i,
\end{split}
\end{align}
(the symmetries \eqref{s:eu} were obtained with the use of the program \emph{MathLie}$^{\text{TM}}$ in \cite{Baumann}). Applying $X_f$ to the alternative Lagrangians $\mc A=\Delta_i,\quad1\le i\le 3$ and using the Noether identity \eqref{nid}, these infinite symmetries lead to the following infinite sets of conservation laws:
\begin{align}\label{c:eu}
\begin{split}
D_t\,\alpha^i_f+D_j\left(\alpha^i_fu^j-\alpha^j_fu^i+\omega^i\beta^j_f-\omega^j\beta^i_f\right)=\pd_jf^i\Delta_i-\pd_tf^i\Delta_4-D_j(f^j\Delta_i).
\end{split}
\end{align}
If we apply now the procedure of Theorem \ref{thm:2} to the conservation laws \eqref{c:eu}, using the fact that functions $g$ are arbitrary functions of all independent variables, we obtain the following (off-shell) differential identities:
\begin{align}\label{id:eu}
D_i\left(D_t\Delta_4-D_j\Delta_j\right)=0,\quad 1\le i\le 3.
\end{align}
Examining the interior expression, and the fact that $D_j\Delta_j=\pd_j\pd_t\omega^j=\pd_t\pd_j\omega^j=D_t\Delta_4$, we find a stronger off-shell identity:
\begin{align}\label{id:eu:1}
D_t\Delta_4-D_j\Delta_j=0.
\end{align}
Thus, the existence of infinite symmetries \eqref{s:eu} with arbitrary functions of all independent variables is an indication that the vorticity system \eqref{e:eu} is under-determined.

\subsection{Symmetries and Noether's second theorem}

Let us try to relate the differential identity \eqref{id:eu:1} (and trivial conservation laws \eqref{c:ch}) to an infinite variational symmetry according to Theorem \ref{thm:noether2}.  

As shown in \cite{Olver}, any differential system can be given a variational formulation by introducing additional unknown variables and equations.  If we start with a differential system $\Delta_a, \:\:  1\le a\le n$ for variables $u^q, \:\:1\le q\le m$, we can introduce additional ``adjoint variables" $v^a, \:\: 1\le a\le n$ and define the following Lagrangian:
\begin{align}\label{l:ol}
\mc L(u,v)=v^a\Delta_a(u).
\end{align}
Indeed, applying the Euler operator $E_{v^a}=\partial_{v^a}+\dots$ to $\mc L$ recovers $\Delta_a$.  The adjoint variables $v^a$ will satisfy their own differential systems, but these functions are, in general, non-physical.

We modify the approach of \cite{Olver} and construct a variational formulation for \eqref{e:ch}.  Letting $\vec\omega=\nabla\times u$ for ``velocity potential" $\vec u$, we rewrite \eqref{e:ch} equivalently as a higher order system for $\vec u$:
\begin{align}\label{e:ch:1}
\vec\Delta=\partial_t(\nabla\times\vec u)+\nabla\times\vec M(t,x,\nabla\times\vec u,...)=0.
\end{align}
We define the following Lagrangian:
\begin{align}\label{l:ch}
\mc L(u,v)=-\vec u\cdot\pd_t\,(\nabla\times\vec v)+\vec M\cdot\nabla\times\vec v.
\end{align}
This Lagrangian recovers \eqref{e:ch:1}, since 
\begin{align*}
\mc L-\vec v\cdot\vec\Delta=-\pd_t\,(\vec v\cdot\nabla\times\vec u)+\nabla\cdot(\vec v\times\vec M-\pd_t\,\vec v\times\vec u),
\end{align*}
which means that $E_{v^a}\mc L=\Delta_a$. 

Now, the Lagrangian \eqref{l:ch} admits the following infinite set of (gauge) symmetries:
\begin{align}\label{s:ch}
\vec u\to \vec u+\nabla F,
\end{align}
where $F(t,x)$ is an arbitrary function of all independent variables.

Applying the Noether identity \eqref{nid}, we associate to this symmetry the following set of conservation laws involving an arbitrary function $F$:
\begin{align}\label{c:ch:2}
\pd_t\,\left(\nabla F\cdot\nabla\times\vec u\right)+\nabla\cdot\left(\vec M\times \nabla F+\nabla(\pd_t\,F)\times\vec u\right)=\nabla F\cdot\vec\Delta.
\end{align}
Rewriting the last term in the LHS of \eqref{c:ch:2}:
\begin{align*}
\nabla(\pd_t\,F)\times\vec u=\nabla\times(\pd_t\,F\,\vec u)-\pd_t\,F\,\nabla\times\vec u,
\end{align*}
and substituting it into \eqref{c:ch:2}:
\begin{align}\label{c:ch:3}
\pd_t\,\left(\nabla F\cdot\nabla\times\vec u\right)+\nabla\cdot\left(\vec M\times \nabla F-\pd_t\,F\,\nabla\times\vec u\right)=\nabla F\cdot\vec\Delta,
\end{align}
for $\vec\omega=\nabla\times\vec u$ recovers the infinite conservation law \eqref{c:ch:1} proposed in \cite{Cheviakov}.

Finally, performing on \eqref{c:ch:3} the integration procedure of Noether's second theorem (see e.g. Theorem \ref{thm:2}) yields the following off-shell identity:
\begin{align}\label{id:ch:1}
\nabla\cdot\vec\Delta=0,
\end{align}
which verifies the fact that the vorticity system \eqref{e:ch:1} is a total curl.  This is unsurprising from a fluid dynamics perspective, since the vorticity system is obtained precisely by taking the curl of the Navier-Stokes system.  We see that the infinite conservation laws \eqref{c:ch} are expressions of this structural fact.

\subsection{Navier-Stokes equations}
In \cite{Cheviakov} were considered special cases of \eqref{c:ch} for physical systems including Maxwell's equations, the Navier-Stokes equations, and the equations of magnetohydrodynamics.  In \cite{Chevober} were considered infinite conservation laws \eqref{c:ch} in the case of the Euler and Navier-Stokes equations of incompressible fluid dynamics; many more results were subsequently derived using these conservation laws.  For Navier-Stokes, the system \eqref{e:ch} for $\omega$ takes the following form:
\begin{align}\label{e:ns}
\begin{split}
\vec\Delta&=\vec \omega_t+\nabla\times\left(\vec\omega\times\vec u-\nu\nabla^2\vec u\right)=0,\\
\Delta_4&=\nabla\cdot\vec \omega=0.
\end{split}
\end{align}
Here, $\vec u$ is the velocity vector, and $\nu$ is the viscosity. The system \eqref{e:ns} is the system \eqref{e:ch} for $\vec M=\vec\omega\times\vec u-\nu\nabla^2\vec u$; it comprises the \emph{vorticity equations} of incompressible flow.  Equation $\Delta_4$ expresses the fact that $\vec\omega=\nabla\times\vec u$.

The infinite conservation laws \eqref{c:ch} considered in these papers take the following form:
\begin{align}\label{c:ns}
\begin{split}
\pd_t(\vec\omega\cdot\nabla F)+\nabla\cdot\left[(\vec\omega\times\vec u-\nu\nabla^2\vec u)\times\nabla F- (\pd_tF)\,\vec\omega\right]&=\nabla F\cdot\vec\Delta-(\pd_tF)\,\Delta_4\\
&=\nabla\cdot(F\vec\Delta)-D_t(F\Delta_4).
\end{split}
\end{align}
As we have shown, these conservation laws are trivial: 
the RHS is a divergence expression with 
fluxes
that vanish on-shell (triviality of the first kind) and, therefore, do not contribute to the conserved densities or integrals.
The rest is a differential identity that holds everywhere in the extended space (off-shell; triviality of the second kind). 

Let us discuss the conserved charge of \eqref{c:ns}, namely $Q[\vec\omega]:=\int\vec\omega\cdot\nabla F\di^3x$, which was proposed in \cite{Cheviakov}. If we integrate \eqref{c:ns} over space and assume that $F$ and $\vec\omega$ are such that the surface terms of the fluxes (arguments of $\nabla\cdot()$) vanish at the boundary (infinity), then we obtain the following relation for $Q$:
\begin{align}\label{Q}
D_tQ[\vec\omega]=-D_t\int F\Delta_4\di^3x, 
\end{align}
meaning that $Q$ is conserved on-shell:
\begin{align}\label{Q}
D_tQ[\vec\omega]\doteq 0. 
\end{align}

Let us now construct an off-shell identity by subtracting the terms on the right hand side of \eqref{c:ns} from both sides:
\begin{align}\label{c:ns:1}
\pd_t(\vec\omega\cdot\nabla F)+D_t(F\Delta_4)+\nabla\cdot\left[(\vec\omega\times\vec u-\nu\nabla^2\vec u)\times\nabla F-\pd_tF\,\vec\omega\right]-\nabla\cdot(F\vec\Delta)\:=\:0,
\end{align}
or
\begin{align}\label{c:ns:2}
\pd_t\,[\nabla\cdot(F\,\vec\omega)]+\nabla\cdot\left[\nabla\times\left(F(\vec\omega\times\vec u-\nu\nabla^2\vec u)\right)-\partial_t(F\,\vec\omega)\right]=0.
\end{align}
The equation \eqref{c:ns:2} is an equivalent statement of \eqref{c:ns}. Let $R[\vec\omega]:=\int\nabla\cdot(F\,\vec\omega)\di^3x$ be the conserved charge of this formulation.  Integrating \eqref{c:ns:2} over space and assuming that all surface terms vanish at the boundary (infinity), we see that $R$ satisfies the following off-shell integral identity:
\begin{align*}
D_tR[\vec\omega]=0,
\end{align*}
meaning that $R$ is conserved off-shell.  In general, when $\vec\omega$ does not satisfy \eqref{e:ns}, $Q\neq R$.
However, for those $\vec\omega$ that satisfy $\Delta[\vec\omega]=0$, assuming the stricter of two boundary conditions, we can demonstrate that $Q=R$ on-shell.  Indeed,
\begin{align*}
Q= \int\vec\omega\cdot\nabla F\di^3x \doteq \int\left(\vec\omega\cdot\nabla F+F\,\nabla\cdot\vec\omega\right)\di^3x =\int\nabla\cdot(F\vec\omega)\di^3x=R.
\end{align*}
If we call $R$ the ``trivial charge", the charge that is conserved for every $\vec\omega$, we see that, although $Q$ is a nontrivial charge off-shell, it is equal to the trivial charge $R$ on-shell.  This is a contrast to conventional conserved quantities for well-defined systems, which are equal to trivial charges neither on- nor off-shell.  

Further study is needed to understand the role of trivial conservation laws
and the nature of conserved charges for under-determined systems; see also \cite{Mintchev86}, \cite{Olver}, \cite{Olver83}.

\subsection{Ertel's Theorem}
Infinite trivial conservation laws \eqref{c:ns} also play a role in inviscid fluid dynamics, $\nu=0$.  They take the form of \emph{Ertel's theorem}, an important tool in atmospheric sciences (see e.g. \cite{Muller}).  The Euler vorticity equations for $\omega$ are as follows:
\begin{align}\label{e:ev}
\begin{split}
\vec\Delta&=\vec \omega_t+\nabla\times\left(\vec\omega\times\vec u\right)=0,\\
\Delta_4&=\nabla\cdot\vec \omega=0,
\end{split}
\end{align}
which is the system \eqref{e:ch} for $\vec M=\vec\omega\times\vec u$, where, as earlier, $\vec u$ is the velocity vector.

Ertel's theorem \cite{Webb2015} states that the following relationship holds on solutions of \eqref{e:ev}:
\begin{align}\label{e:er}
\partial_t(\vec\omega\cdot\nabla\psi)+\nabla\cdot[(\vec\omega\cdot\nabla\psi)\vec u]=0.
\end{align}
Here, $\psi(t,x)$ is a scalar function advected with the flow of $\vec u$.  This means that it solves the following differential equation:
\begin{align}\label{psi}
\partial_t\psi+\vec u\cdot\nabla\psi=0.
\end{align}

In fact, Ertel's theorem \eqref{e:er} follows directly from trivial conservation law \eqref{c:ns} for $\nu=0$.  To see this, we start with \eqref{c:ns} for $\nu=0$ and $F=\psi$:
\begin{align}\label{c:er}
\partial_t(\vec\omega\cdot\nabla\psi)+\nabla\cdot[(\vec\omega\times\vec u)\times\nabla\psi-\partial_t\psi\,\vec\omega]=\nabla\cdot(\psi\vec\Delta)-D_t(\psi\,\Delta_4).
\end{align}
Since $(\vec\omega\times\vec u)\times\nabla\psi=(\vec\omega\cdot\nabla\psi)\vec u-(\vec u\cdot\nabla\psi)\vec\omega$, equation \eqref{c:er} becomes:
\begin{align*}
\partial_t(\vec\omega\cdot\nabla\psi)+\nabla\cdot[(\vec\omega\cdot\nabla\psi)\vec u-(\pd_t\psi+\vec u\cdot\nabla\psi)\,\vec\omega]=\nabla\cdot(\psi\vec\Delta)-D_t(\psi\,\Delta_4).
\end{align*}
But since $\psi$ satisfies \eqref{psi}, Ertel's theorem \eqref{e:er} follows:
\begin{align*}
\partial_t(\vec\omega\cdot\nabla\psi)+\nabla\cdot[(\vec\omega\cdot\nabla\psi)\vec u]=\nabla\cdot(\psi\vec\Delta)-D_t(\psi\,\Delta_4),
\end{align*}
which is a trivial conservation law, since its characteristic is zero.

\section{Conclusion}
We considered quasi-Noether systems, a class of differential system that includes all equations possessing conservation laws. We extended an approach based on the Noether operator identity and formulated and proved an extension of a Second Noether theorem for quasi-Noether systems. As in the case of equations of a variational problem, the existence of an infinite symmetry group of the system with an arbitrary function of \emph{all} independent variables generates a differential identity between the equations of the system and their derivatives.   
In addition, we showed that 
infinite conservation laws 
involving
an arbitrary function of \emph{all} independent variables are necessarily trivial. We analyzed recently obtained sets of infinite conservation laws (with an arbitrary function of all independent variables) for the Euler and Navier-Stokes equations in vorticity formulations and demonstrated that these conservation laws are trivial.
We also showed that the existence of infinite conservation laws with an arbitrary function of \emph{not all} independent variables leads to the conclusion that these conservation laws are, in general, non-trivial, but that their fluxes are total divergences on solutions.

\end{document}